\documentclass[11pt]{amsart}

\usepackage[USenglish]{babel}

\usepackage{amsmath,amsthm,amssymb,amscd}
\usepackage{booktabs}
\usepackage{url}

\usepackage{MnSymbol}

\usepackage{tikz}
\usetikzlibrary{arrows,shapes.geometric,positioning,decorations.markings}

\usepackage[mathscr]{eucal}
\usepackage[normalem]{ulem}
\usepackage{latexsym,youngtab}
\usepackage{multirow}
\usepackage{epsfig}
\usepackage{parskip}

\usepackage[margin=2cm]{geometry}
\usepackage{color}

\newtheoremstyle{custom}
  {3pt}
  {3pt}
  {\slshape}
  {}
  {\bfseries}
  {.}
  { }
   {}
\theoremstyle{custom}
\newtheorem{theorem}{Theorem}[section]
\newtheorem{proposition}[theorem]{Proposition}

\newtheorem{proposition/definition}[theorem]{Proposition/Definition}

\theoremstyle{definition}

\theoremstyle{remark}
\newtheorem{remark}[theorem]{Remark}

\newcommand{\stack}[2]{\ensuremath{\genfrac{}{}{0pt}{}{#1}{#2}}} 

\newtheoremstyle{exercise}
  {3pt}
  {6pt}
  {}
  {}
  {\bfseries}
  {:}
  { }
   {}
\theoremstyle{exercise}
\newtheorem{exercise}[theorem]{Exercise}
\newtheoremstyle{exercises}
  {3pt}
  {6pt}
  {}
  {}
  {\bfseries}
  {:}
  {\newline}
   {}

\theoremstyle{exercise}
\newtheorem{exercises}[theorem]{Exercises}

\input epsf
\def\boxit#1{\vbox{\hrule height1pt\hbox{\vrule width1pt\kern3pt
  \vbox{\kern3pt#1\kern3pt}\kern3pt\vrule width1pt}\hrule height1pt}}

\def\BR{\mathbb R}

\def\tdim{{\rm dim}}

\def\hd{,...,}

\newcommand{\abs}[1]{\lvert#1\rvert}

\def\11{\mathbf 1}

\def\s{\sigma}

\def\d{\delta}

\def\ot{{\mathord{ \otimes } }}

\def\otc{{\mathord{\otimes\cdots\otimes} }}

\def\La#1{\Lambda^{#1}}

\def\s{\sigma}

\def\FS{\mathfrak  S}

\def\ol{\overline}

\def\BR{\mathbb  R}

\def\hd{, \hdots ,}

\def\La#1{\Lambda^{#1}}

\def\tdeg{\operatorname{deg}}

\def\tdim{\operatorname{dim}}

\def\be{\begin{equation}}
\def\ene{\end{equation}}

\DeclareMathOperator\tspan{span}

\def\tspan{{\rm span}}

\def\bx{{\bold x}}

\DeclareMathOperator{\sect}{Sec}

\renewcommand{\d}{\delta}

\def\bx{\bold x}

\def\La#1{\bigwedge^{#1}}

\usepackage[utf8]{inputenc}
\usepackage{amsmath}
\usepackage{amsthm}
\usepackage{amsfonts}
\usepackage{indentfirst}
\usepackage{xcolor}

\usepackage[sort]{cite}

 \newcommand{\mc}{\mathcal}

\newcommand{\average}[1]{\left\langle#1\right\rangle}
\newcommand{\jbra}[1]{\langle#1\rvert}
\newcommand{\jket}[1]{\lvert#1\rangle}
\newcommand{\jbraket}[2]{\langle#1\vert#2\rangle}

\begin{document}

\author{Hang Huang,  J. M. Landsberg, and Jianfeng Lu}

\address[H. Huang, J.M. Landsberg]{Department of Mathematics, Texas A\&M University, College Station, TX 77843-3368, USA}
\email[H. Huang]{huanghang1109@gmail.com}
\email[J.M. Landsberg]{jml@math.tamu.edu}

\address[J. Lu]{Departments of Mathematics, Physics, and Chemistry, Duke
  University, Durham, NC 27708, USA}
\email[J. Lu]{jianfeng@math.duke.edu}

\title[Geometry of Backflow Ansatz]{Geometry of backflow transformation ansatz for quantum many-body fermionic wavefunctions}


\thanks{Landsberg supported by NSF grant AF-1814254. Lu supported in part by NSF grants DMS-2012286 and CHE-2737263. 
This project is an outcome of the IPAM program: {\it Tensor Methods and Emerging Applications to the Physical and Data Sciences}
  March 8, 2021 -- June 11, 2021}


\begin{abstract}
  Wave function ansatz based on the backflow transformation are widely
  used to parametrize anti-symmetric multivariable functions for
  many-body quantum problems. We study   the geometric
  aspects of such ansatz, in particular we show that in general
  totally antisymmetric polynomials cannot be efficiently represented
  by backflow transformation ansatz at least in the category of
  polynomials. In fact, one needs a linear combination of at least 
  $O(N^{3N-3})$ determinants to represent a generic totally
  antisymmetric polynomial.  Our proof is based on  bounding the dimension of the source of the ansatz
  from above and bounding the dimension of the target from below.
\end{abstract}

\maketitle

\section{Introduction}

Finding efficient numerical methods for quantum many-body systems has
been a long standing challenge, due to the notorious curse of
dimensionality and the Fermionic sign problems. Many approaches have been proposed
over the years; a popular class of methods is known as the variational
quantum Monte Carlo methods. The basic idea is as follows: let $H$
denote the Hamiltonian operator of a quantum system, choose a class of
functions $\mc{F}$ as a variational ansatz and solve  
\begin{equation}
    E_{\mc{F}} := \inf_{\Psi \in \mc{F}} \frac{\jbra{\Psi} H \jket{\Psi}}{\jbraket{\Psi}{\Psi}} 
\end{equation}
for $\Psi$ and $E_{\mc{F}}$. 
For Fermionic systems, due to the Pauli  exclusion principle, the wave function has to be totally antisymmetric. Thus, for a system with $N$ particles one takes $\mc{F} \subset \bigwedge^{N} L^2(\BR^3)$, so that the above $E_{\mc{F}}$ gives a variational upper bound for the true ground state energy. Here $\bigwedge^{N} L^2(\BR^3)$ denotes the anti-symmetric tensor product of $n$ copies of $L^2(\BR^3)$; the single particle Hilbert space is $L^2(\BR^3)$, i.e., we assume that each particle lives in $\BR^3$ and have neglected the spin degree of freedom for simplicity.

The question now becomes the choice of $\mc{F}$. The most
straightforward approach is to take
$\mc{F} = \bigwedge^{N} L^2(\BR^3)$, or more precisely, for numerical
purposes, one chooses a finite dimensional subspace
$V \subset L^2(\BR^3)$ and take $\mc{F} = \bigwedge^{N} V$, in the
spirit of Galerkin's method in numerical analysis. This is however not
practical for actual computations as the dimension of $\mc{F}$ grows
exponentially with $N$, known as the curse of dimensionality.
  Therefore, a
smaller class of functions needs to be fixed for the variational
search.

The most well-known choice of $\mc{F}$, which is essentially the starting point of quantum chemistry, is the collection of Slater determinants, i.e., 
\begin{equation}
    \mc{F} = \bigl\{  \varphi_1 \wedge \varphi_2 \wedge \cdots \wedge \varphi_N \mid \varphi_i \in L^2(\BR^3),  \average{\varphi_i, \varphi_j}_{L^2(\BR^3)} = \delta_{ij}, i, j = 1, \cdots, n \bigr\}. 
\end{equation}
This leads to the celebrated Hartree-Fock method. In practice one also
uses a finite dimensional approximation to $L^2(\BR^3)$.  While useful
as a first approximation, the class of Slater determinants is often
too small to capture a good approximation to the ground state
energy. The difference between the result of the Hartree-Fock method
and the true ground state energy is called the correlation energy;
more complicated variational ansatz have been proposed to reduce the
error.

One of the first approaches to go beyond the Hartree-Fock ansatz is known as the Slater-Jastrow wave function, for which one considers the product of a Slater determinant with a totally symmetric function $g$ and hence the product is anti-symmetric. The function $g$ is often parametrized as 
\begin{equation}
  g(x_1, \cdots, x_N) = \exp \Bigl( \sum_{i \leq j } U(\abs{x_i- x_j}) \Bigr),
\end{equation}
where $U$ is some function on $\BR$. The $g$ given above is obviously
totally symmetric, while more general ansatz for $g$ have also  been  
proposed and studied. Unlike the Hartree-Fock method, it is no longer
possible to explicitly evaluate the Rayleigh quotient, and thus the
ansatz is optimized by Monte Carlo approaches in practice, i.e.,
variational quantum Monte Carlo methods.  One 
generalization   uses linear combinations of Slater-Jastrow
wave functions, which is often referred to as the
multi-configurational approach in the  variational
quantum Monte Carlo methods   literature.

Another direction is to change the Slater determinants to some other anti-symmetric functions; for example, pfaffians (when $N$ is even), Vandemonde determinants, and determinants with backflow transformations as defined below. 
This has become a very active field in recent years thanks to the rise of neural networks as
a  versatile ansatz for high dimensional functions, after the influential work \cite{carleo2017solving} of parameterizing many-body wave functions using neural networks. Several variational classes have been proposed by replacing components in the anti-symmetric function ansatz by neural networks, see e.g., \cite{luo2019backflow, choo2020fermionic,
han2019solving, hermann2020deep, pfau2020abinitio, acevedo2020vandermonde, hutter2020representing, spencer2020better, schatzle2021convergence}. 
While the details of these ansatz differ, the general framework is based on the backflow transformation originally proposed by Feynman and Cohen \cite{feynman1956energy}. 
To introduce the ansatz,  define the function class 
\begin{equation*}
    \mathcal{S} = \bigl\{ \varphi \in L^2\bigl(\mathbb{R}^3 \times (\mathbb{R}^3 \otimes \mathbb{R}^{N-1}) \bigr) \mid 
    \varphi(\bx; \ol{\boldsymbol{y}}) = \varphi(\bx; \sigma \ol{\boldsymbol{y}}), \quad \forall \sigma \in \FS_{N-1} \bigr\}.
\end{equation*}
Thus, functions in $\mathcal{S}$ are totally symmetric with respect to the second argument in $\mathbb{R}^3 \otimes \mathbb{R}^{N-1}$ (the permutation group $\FS_{N-1}$ acts on $\mathbb{R}^{N-1}$).
It is easy to check then for $\varphi_i \in \mathcal{S}$, $i = 1, \cdots, N$, the following function on $\mathbb{R}^3 \otimes \mathbb{R}^N$ is totally anti-symmetric: 
\begin{equation*}
    \Phi_{\text{BF}}[\varphi_1, \ldots, \varphi_N](\bx_1, \cdots, \bx_N) = \det 
    \begin{pmatrix}
        \varphi_1(\bx_1; \ol{\bx_{-1}}) & \cdots & \varphi_N(\bx_1; \ol{\bx_{-1}}) \\
        \vdots & \ddots & \vdots \\
        \varphi_1(\bx_N; \ol{\bx_{-N}}) & \cdots & \varphi_N(\bx_N; \ol{\bx_{-N}}) \\
    \end{pmatrix}
\end{equation*}
where $\ol{\bx_{-i}}:=(\bx_1, \cdots, \bx_{i-1}, \bx_{i+1}, \cdots, \bx_N)$. This generalizes the Slater determinants, which corresponds to the case that $\varphi_i$ only depends on the first variable, the Slater-Jastrow wave functions, which corresponds to absorbing a $g^{1/N}$ factor into $\varphi_i$'s,  and the Vandermonde determinants.  
We call $\Phi_{\text{BF}}$ the {\it ansatz map}.

In this work, we focus on the study of the ansatz given by the backflow transformation:
\begin{equation*}
    \mc{F}_{\text{BF}}^N = \bigl\{ \Phi_{\text{BF}}[\varphi_1, \ldots, \varphi_N] \mid 
    \varphi_i \in \mc{S}, i = 1, \ldots, N \bigr\}.
\end{equation*}
We are interested in the representation power of the  class of functions
$ \mc{F}_{\text{BF}}^N$. In particular, we ask for a given totally antisymmetric function $\Psi: \mathbb{R}^{3} \otimes \mathbb{R}^N \to \mathbb{R}$, whether it belongs to $\mc{F}_{\text{BF}}^N$, i.e., whether it is possible to find $\{\varphi_i\}$, such that $\Psi = \Phi_{\text{BF}}[\varphi_1, \ldots, \varphi_N]$. 
This question may  be asked in several different flavors, depending on the assumed function classes of $\Psi$ and 
the $\{\varphi_i\}$. In this work, we approach the question from the perspective of algebra. Thus, we consider $\Psi$ and hence
the  $\{\varphi_i\}$  to be   polynomials. 

Let $S^\d\BR^{3*}$ denote the space of homogeneous degree $\d$ polynomials on $\BR^3$.
Consider
\begin{equation*}
  \La N(L^2({\BR^3}))_{alg}:=\La N(\bigoplus_{\d=0}^{\infty} S^\d {\BR^3}^*),
\end{equation*}
where for any given element of $L^2({\BR^3})$, 
we only allow a finite number of $\d$ to be used.  
  Fix a total degree $D$ and
consider $\La N(L^2({\BR^3}))_{alg}\cap {\BR^3}^{*\ot D}$ and call this the $(N,D)$-{\it space}.
Below we will see that for a given $N$, one must have $D$ at least on the order of $N^{\frac 43}$. 

Our main result states that in general totally antisymmetric   polynomials do not belong to $\mc{F}_{\text{BF}}^N$.
We stratify the set  $\mc{F}_{\text{BF}}^N$ by total degree $D$
and write   $\mc{F}_{\text{BF}}^N=\bigoplus_D
\mc{F}_{\text{BF}}^{D,N}$. Then  $ \mc{F}_{\text{BF}}^{D,N}$ is an algebraic subvariety of the $(N,D)$-space
$\La N(L^2({\BR^3}))_{alg}\cap {\BR^3}^{*\ot D}$. Note that the elements
mapping to $\mc{F}_{\text{BF}}^{D,N}$ are the $(\varphi_1\hd \varphi_N)$
such that  $\sum_{j=1}^N\tdeg(\varphi_j)=D$ and each $\varphi_j$ is homogeneous.

\begin{theorem}\label{mainthm} For each fixed $N$, for all $D$ sufficiently large,
the algebraic ansatz map is not surjective. The dimension of the target is
roughly $N^{3N-3}$ times larger than the dimension of the source.

In particular, when $D$ is sufficiently large, for all $r<N^{3N-3}$, the set of  sums of $r$ elements of  
$\mc{F}_{\text{BF}}^{D,N}$ still
 lies  in a proper subvariety of $\La N(L^2({\BR^3}))_{alg}\cap {\BR^3}^{*\ot D}$. In particular, it is a set 
 of measure zero.
\end{theorem}

Theorem \ref{mainthm} will follow immediately from the upper bound on the dimension of the source
in \S\ref{source} and the lower bound on the dimension of the target in \S\ref{target}.

\begin{remark} Our estimates are coarse, but they only  assume $D> N^3$. The map will still fail to be surjective
for all but very few admissible values of $D$.
\end{remark}

\begin{remark} The ansatz map has positive dimensional fibers. It would be interesting to determine their dimensions.
\end{remark}

\begin{remark} The closure of the image of the ansatz map is some variety invariant under the action
of $GL_3$. It would be interesting to obtain geometric information about this variety. 
\end{remark}

Our theorem, to some extent, is a negative result for the
representation power of the backflow transformation ansatz, at least
in the category of polynomials, since the number of elements needed
grows exponentially in $N$. Note that any analytic
function is a limit of a sequence of polynomials. Restriction to
homogeneous polynomials is not a restriction as we may always project
to homogeneous components, which will be the images of homogeneous
$\phi_i$.

We remark that recent work \cite{hutter2020representing} argues that
any totally anti-symmetric function can be represented by the backflow
ansatz (see \cite[Theorem 7]{hutter2020representing}), however the
$\varphi_i$ used in the construction involve a sorting of coordinates
$\bx_i$ in ``lexicographical'' order, and are hence discontinuous and
also impractical for actual computations.

\subsection*{Organization} In \S\ref{prelim} we review standard results needed for the proof.
 In \S\ref{target} we bound the dimension of the target from below and
in \S\ref{source} we bound the dimension of the source from above. The two estimates
together prove Theorem \ref{mainthm}. We conclude in \S\ref{geom} with geometric remarks.

\subsection*{Acknowledgement}
This project is an outcome of the IPAM program {\it Tensor Methods and
  Emerging Applications to the Physical and Data Sciences} March 8,
2021 to June 11, 2021. The authors thank IPAM for bringing us
(virtually) together and Landsberg thanks the Clay foundation for the
opportunity to serve as a senior scholar during the program. 
   
 \section{Preliminaries}\label{prelim}
 Let  $\ol p_{k}(m)$ denote the number of partitions of $m$ with at most $k$ parts.
 The following estimates are standard. We include proofs for the convenience of the reader.

\begin{proposition} \label{pkprop}
 $\ol p_k(m)= \frac 1{k!(k-1)!} m^{k-1} + O(m^{k-2}).$
\end{proposition} 
\begin{proof} 
Recall Faulhaber's formula:  $1^k+2^k+\cdots + m^k=\frac{m^{k+1}}{k+1}+O(m^k)$.

We have the induction formula 
$$
    \ol p_k (m)  = \sum_{i = 0}^{\lfloor \frac{m}{k} \rfloor} \ol p_{k-1}(m - ik) .
$$
When $k = 1$, we have $\ol p_k(m) = 1$.
Assume by induction that 
$\ol p_u(m)=c_u m^{u-1}+O(m^{u-2})$ for all $u<k$.
We prove it holds for $k$ and   $c_u=\frac 1{u!(u-1)!}$.
\begin{align*}
    \ol p_k (m) &= \sum_{i = 0}^{\lfloor \frac{m}{k} \rfloor} \ol p_{k-1}(m - ik) \\
    & = \sum_{i = 0}^{\lfloor \frac m{k} \rfloor} c_{k-1} (m-ik)^{k-2} + O((m-ik)^{k-3}) \\
    &  = k^{k-2}c_{k-1}\sum_{i = 0}^{\lfloor \frac m{k} \rfloor}(\frac mk-i)^{k-2}+ O((m-ik)^{k-3})\\
    &=  k^{k-2} c_{k-1}\frac{1}{k-1}(\frac mk)^{k-1}+O(m^{k-2})\\
    &= \frac 1{k (k-1) }c_{k-1} m^{k-1}+O(m^{k-2}).
      \end{align*} 
We conclude by induction.  
\end{proof}

\begin{proposition} Let $\ol q_k(m)$ denote the number of  partitions of 
$m$ with either  $k$ or $k-1$ parts that are strictly decreasing. Then  $\ol q_k(m+\binom{k}2)=\ol p_k(m)$.
\end{proposition}
\begin{proof} Given a partition of $m$ with at most $k$ parts, we may obtain a new   partition that is strictly decreasing
by adding one to the $(k-1)$-st part, two to the $(k-2)$-nd, up to $(k-1)$ to the first.
Moreover all   strictly decreasing partitions with $k$ or $k-1$ parts arise in this way so we have established a bijection.
\end{proof}

 Recall from Stirling's formula that 
 $n!=\sqrt{2\pi n}(\frac ne)^n(1+O(\frac 1n))$ so in particular
 $ n!(n-1)! =2\pi e^{-2n+1}n^{2n}(1+O(\frac 1n))$.
 
 In what follows we will be concerned with estimates so we suppress round ups and round downs to
 integers from the notation.

 A basic fact from algebraic geometry is that the dimension of a finite union of algebraic
varieties is the dimension of the largest component.
Because of this we may restrict to homogeneous polynomials and a single
source.

\section{Target space}\label{target}

We fix a total degree $D$ and lower bound the dimension of
$$
\bigoplus_{\stack{0\leq p_1\leq p_2\leq \cdots\leq p_N}{p_1+\cdots +p_N=D}}
S^{p_1}W^*\otc S^{p_N}W^*.
$$

We only consider terms where
 $p_1\geq D/2N$ and $p_1<p_2<\cdots <p_N$. The sum becomes
 $$
 \sum_{\stack{q_1+\cdots + q_N=D/2}{ q_1< q_2<\cdots < q_N}}
 S^{q_1+\frac D{2N}}W^*\ot \cdots \ot S^{q_N+\frac D{2N}}W^*.
 $$
  Recall that $\tdim S^dW^*=\binom{d+1}2$. 
 The smallest term in the summation is when $q_j=j-1$ for $j<N$ and $q_N=\frac{D}{2 } -\binom{N-1}2$.
 Assume $D>N^3$. 
 This term has
 dimension
 $$
\binom{\frac{N+1}{2N}D-\binom{N-1}2+2}2 \prod_{j=0}^{N-2} \binom{j+\frac D{2N}+2}2
=\frac{D^{2N}}{2^{3N}N^{2(N-1)}}+O(D^{2N-1})
$$
 
 The number of terms is 
 $$
 \ol q_N(\frac D2)=\frac 1{N!(N-1)!}(\frac D2+\frac{N(N-1)}2)^{N-1}+O((\frac D2+\frac{N(N-1)}2))^{N-2}
 =\frac{D^{N-1}e^{2N-1}}{\pi 2^{N }N^{2N }}+O(D^{N-2})
 $$ 
 
 Thus the dimension of the target is bounded below by
 
  $$
 [ \frac{D^{N-1}e^{2N-1}}{\pi 2^{N }N^{2N }}+O(D^{N-2})][\frac{D^{2N}}{2^{3N}N^{2(N-1)}}+O(D^{2N-1})]
 =
 \frac{D^{3N-1}e^{2N-1}}{\pi 2^{4N}N^{4N-2}}+O(D^{3N-2})
 $$

\section{Source space}\label{source}

Let $d_j=\tdeg\varphi_j$, then  $\mc{F}_{\text{BF}}^{D,N}$
is the union  of the images of the  ansatz maps over all $(\varphi_1\hd \varphi_N)$ with $d_1\leq \cdots \leq d_N$
and $\sum d_j=D$.  By the basic fact in algebraic geometry
mentioned above, the ansatz surjects onto the $(N,D)$ space if and only if one of the
$(d_1\hd d_N)$-ansatz maps does.  

Fix $d=d_j$ and suppress the $j$ index. Write
$$\varphi=\varphi(\bx_1;\ol{\bx_{-1}})=\sum_{z=0}^d\sum_{\stack{\d_2+\cdots + \d_N=d-z}{0\leq \d_2\leq \d_3\leq \cdots \leq \d_N }}  f^{z,\ol\d}_z(\bx_1) 
\sum_{\s\in \FS_{N-1}}
h^{z,\ol\d}_2( \bx_{\s(2)})\cdots h^{z,\ol\d}_N( \bx_{\s(N)})
$$
Here $\ol\d=(\d_2\hd \d_N)$,  $f^{z,\ol\d}_z$ has degree $z$, and $h^{z,\ol\d}_j$ has degree $\d_j$.
We are assuming without loss of generality that the $h$'s  are non-decreasing in degree with $j$
because they appear symmetrically. On the other hand, we have to allow the degree of the $f$'s to be
any   value from $0$ to $d$.

 The dimension of the source $\varphi=\varphi_j$ with $j$ fixed and $d=d_j$ is thus
$$
 \sum_{z=0}^d \sum_{\stack{\d_2+\cdots + \d_N=d-z}{0\leq \d_2\leq \d_3\leq \cdots \leq \d_N }}  
[\binom{z+2}2 \binom{\d_2+2}2 \cdots   \binom{\d_N+2}2]
$$
 where   $\binom{z+2}2$ is the dimension of the space of  $f_z$'s of degree $z$ and 
 $\binom{\d_j+2}2$ is the dimension of the space of $h_{\d_j}$'s of degree $\d_j$.
 
 The largest component of the source will be when $d_j=\frac DN=:d$ for all $j$. 
 We have the following upper bound for the dimension of the source:
 \begin{align*}
 &N\sum_{z=0}^d \sum_{\stack{\d_2+\cdots + \d_N=d-z}{0\leq \d_2\leq \d_3\leq \cdots \leq \d_N }}  
[\binom{z+2}2 \binom{\d_2+2}2 \cdots   \binom{\d_N+2}2]\\
&<N \sum_{z=0}^d   
 \binom{z+2}2 {\binom{\frac{d-z}{N-1}+2}2}^{N-1} \ol p_{N-1}(d-z)\\
 &<
 Nd{\binom{\frac{d}{N}+2}2}^{N}  \ol p_{N-1}(d)\\
 &=Nd
 \frac{d^{2N}}{2^{N}N^{2N}}
 \frac{d^{N-2}e^{2N-3}}{2\pi N^{2N-2}}+O(d^{3N-2})\\
 &=
 \frac{D^{3N-1}e^{2N-3}}{N^{7N-4}2^{N+1}\pi}+O(D^{3N-2})
 \end{align*}
 
 The second line holds because the product of the binomial coefficients with the $\d_i$'s 
 is largest when they are all equal. The third holds because the largest term in the second line occurs
 when $z=\frac dN$, and the last by using approximations to the terms.

 \section{Geometric Discussion}\label{geom}
 Given a vector space $V$ and an algebraic subset $X\subset V$, one can 
 define the variety of points on $r$-secant planes of $X$:
 $$
 \sect_r(X):=\ol{\{ v\in V\mid \exists x_1\hd x_r\in X, v\in \tspan \{ x_1\hd x_r\}\} }
 $$
 For any $X$, a na\"\i ve dimension count gives 
 \be\label{secdim}\tdim \sect_r(X)\leq r\tdim(X)+r,
 \ene
  as one chooses
 $r$ points on $X$ and a point in their span. These secant varieties (or more precisely, their
 cousins in projective space) are intensely studied in algebraic geometry.
 
  In our case a natural question is to take $X$ to be the image of an ansatz map  
 and to ask when its secant varieties fill the ambient space.

\bibliographystyle{plain}
\bibliography{ref.bib}
\end{document}